\documentclass[a4paper,11pt]{article}
\usepackage{amssymb}
\usepackage[latin1]{inputenc}
\usepackage[british]{babel}
\usepackage{natbib}
\usepackage{graphicx}
\usepackage[pdftex]{lscape}
\usepackage{verbatim}
\usepackage{amsmath}
\usepackage{arydshln}
\usepackage{amsthm}
\usepackage{amssymb}
\usepackage{amsfonts}
\usepackage{setspace}
\usepackage{latexsym}
\usepackage{a4}
\usepackage{bbm}
\usepackage{mathrsfs}
\usepackage{fixmath}
\usepackage{tikz}   
\usepackage{subfigure}
\usepackage{float}
\usetikzlibrary{fit}
\usepackage{hyperref}

\newcommand{\rk}{\mathrm{rank}\,}

\newcommand{\Chol}{\Sigma_{tr}}

\newcommand{\Qtj}{\tilde{Q}_j}
\newcommand{\fA}{f(A_0,A_+)}
\newcommand{\jn}{j=1,\ldots,n}

\newcommand{\e}{\varepsilon}

\newcommand{\Atot}{\left(A_0,A_+\right)}

\newcommand{\Vj}{\mathcal{V}_j}
\newcommand{\Vjo}{\mathcal{V}_j^\bot}
\newcommand{\sPi}{\mathcal{P}_i}
\newcommand{\sPio}{\mathcal{P}_i^\bot}
\newcommand{\sH}{\mathcal{H}}

\newcommand{\PRt}{\tilde{\ensuremath\mathbb{P}}^R}
\ifdefined\O
    \renewcommand{\O}{\ensuremath\mathcal{O}\left(n\right)}
\else
    \newcommand{\O}{\ensuremath\mathcal{O}\left(n\right)}
\fi
\ifdefined\Re
    \renewcommand{\Re}{\ensuremath\mathbb{R}}
\else
    \newcommand{\Re}{\ensuremath\mathbb{R}}
\fi
\ifdefined\A
    \renewcommand{\A}{\ensuremath\mathcal{A}_0(\phi)}
\else
    \newcommand{\A}{\ensuremath\mathcal{A}_0(\phi)}
\fi
\ifdefined\Ar
    \renewcommand{\Ar}{\ensuremath\mathcal{A}_0^r(\phi)}
\else
    \newcommand{\Ar}{\ensuremath\mathcal{A}_0^r(\phi)}
\fi
\ifdefined\dim
    \renewcommand{\dim}{\mathrm{dim}\,}
\else
    \renewcommand{\dim}{\mathrm{dim}\,}
\fi

\def\thesection{\Roman{section}}
\makeatletter

\newtheorem{theo}{Theorem}
\newtheorem{prop}{Proposition}
\newtheorem{lemma}{Lemma}
\newtheorem{algo}{Algorithm}
\theoremstyle{definition}
\newtheorem{defin}{Definition}
\newtheorem{condition}{Condition}
\newtheorem{ex}{Example}

\date{}

\begin{document}

\title{\huge On global identification in structural vector autoregressions\thanks{We thank Thomas Carr, Giuseppe Cavaliere, Luca Fanelli, Michele Piffer, Majid Al-Sadoon and Matthew Read for beneficial discussions and comments. Financial support from the ESRC through the ESRC Centre for Microdata Methods and Practice (CeMMAP)
(grant number RES-589-28-0001) and the European Research Council (Starting grant No. 715940) is gratefully acknowledged. Emanuele Bacchiocchi gratefully acknowledges financial support
from Italian Ministry of University and Research (PRIN 2022, Grant 20229PFAX5) and the University of Bologna (RFO grants). The data that support the findings of this study are openly available in [repository name] at [URL], reference number [reference number].}}

\author{%
\begin{tabular}{ccc}
Emanuele Bacchiocchi\thanks{University of Bologna, Department of Economics. Email: e.bacchiocchi@unibo.it} &  & 
Toru Kitagawa\thanks{Brown University, Department of Economics. Email: toru\_kitagawa@brown.edu} \\ 
University of Bologna &  & Brown University \\ 
\end{tabular}%
}


\date{This draft: 7 July 2026}
\maketitle

\begin{abstract}
In a landmark contribution to the structural vector autoregression (SVARs) literature,  Rubio-Ram\'{i}rez, Waggoner, and Zha (2010, `Structural Vector Autoregressions: Theory of Identification and Algorithms for Inference,' \textit{Review of Economic Studies}) show necessary and sufficient conditions for equality restrictions to globally identify the structural parameters of an SVAR. Among them, the sufficient condition shown in their Theorem 7 is the simplest and most attractive for practitioners, reducing the check for global identification to a counting exercise about the number of zero restrictions imposed. 
However, their findings build on a set of regularity assumptions, one of which states that
the first derivative of the function transforming the parameters to be constrained must have full rank. We show, through empirically relevant examples, this assumption can fail in practice and, if so, how naive applications of the necessary and sufficient condition shown in their Theorem 7 fails to check global identification. We derive a modified and more general necessary and sufficient condition for SVAR global identification and show how it can be easily applied in practice.
\end{abstract}


\begin{flushleft}
\end{flushleft}

\newpage


\section{Introduction}
\label{sec:intro}

\cite{RWZ10} (henceforth RWZ) provide necessary and sufficient conditions for the global identification of structural parameters in Structural Vector Autoregressions (SVARs) under a general class of zero restrictions imposed on the structural parameters and their (non-)linear transformations, including impulse responses. Exploiting the insights of their global identification analysis, RWZ also develop efficient and practical  algorithms to perform estimation and inference for structural parameters and impulse responses. Their analytical and computational innovations have been instrumental to recent developments in the literature, including set-identified SVARs (\citet{ARW18}, \citet{GK18}, \citet{GKR19}, \citet{Volpicella20}, \citet{AD21}), locally-identified SVARs (\citet{BK20}), and SVARs with narrative restrictions (\citet{AR18}, \citet{GKR21}), to list a few. RWZ provide several different versions of the necessary and sufficient conditions for global identification. The one given in Theorem 7 is the simplest and most attractive for practitioners, which reduces the check for global identification to a counting exercise about the number and pattern of imposed zero restrictions without requiring knowledge of the true value of the structural or reduced-form parameters. For instance, in \citet{ACR19} and \citet{Zviadadze17}, the authors apply Theorem 7 of RWZ to judge whether the imposed identifying restrictions deliver global identification or not.

However, their rank condition builds on some technical assumptions that should be checked in advance. Through a reasonable example, we show the importance of such 
assumptions and to what extent their failure can lead to missleading results, making the rank condition no longer sufficient. An analytical investigation of this example reveals why it does not guarantee global identification. We find that the condition of Theorem 7 of RWZ cannot detect what we refer to as \textit{redundancy} of imposed identifying restrictions. In this phenomenon, a set of equality restrictions on the structural parameters or impulse responses implicitly forces other structural parameters or impulse responses to zero. If it is present, some (redundant) zero restrictions are already implied by other imposed equality restrictions, so they do not contribute any further identifying information to the system. The condition of Theorem 7 of RWZ, without a prior check of the basic technical assumptions, incorrectly counts the redundant identifying restrictions as if they reduced the dimension of the admissible structural parameters, resulting in an erroneous conclusion that the model is globally identified. We argue that the redundancy of the identifying restrictions is relevant for empirical applications, rather than being of pure theoretical interest.

In the present paper we provide a new necessary and sufficient condition for (exact) global identification that correctly discounts redundant identifying restrictions, and can be used under weaker assumptions than in RWZ. RWZ propose a useful algorithm that sequentially constructs an orthonormal matrix for structural parameter identification that satisfies the identifying restrictions. Building on and modifying their algorithm, our proposed necessary and sufficient condition for global identification checks for the existence of redundant restrictions by verifying whether the orthonormal matrix generated by this sequential algorithm is unique. Verifying uniqueness boils down to checking the rank of a sequence of matrices constraining each column of the orthonormal matrix. Our algorithm only requires values of the reduced-form parameters as an input and checks for global identification at that particular point of the parameter space. Our approach is simpler in practice than analytically checking a priori the theoretical assumptions in RWZ, and, furthermore, applicable to cases even when they violate their regularity assumptions. 

As an alternative to their Theorem 7, Theorem 1 in RWZ presents a different form of necessary and sufficient condition for global identification. As we illustrate in this paper, its proper implementation requires a complete understanding of how the imposed identifying restrictions analytically constrain the impulse responses and the set of structural parameters. For instance, if redundant identifying restrictions are present but one is not aware which zero restrictions can be implied by others, naive implementation of the rank conditions in Theorem 1 of RWZ may also overlook a lack of global identification. To prevent this, it is important to analytically ascertain how a set of equality restrictions translates to zero restrictions for other structural objects. This is feasible for small scale SVARs, but can be less straightforward for medium or large scale SVARs. In contrast, checking our necessary and sufficient condition remains tractable and attractive even for moderate- to large scale SVARs. 

We also show, through an empirical implementation of our methodology, to what extent it can find fertile ground in the growing literature on proxy SVARs.

The remainder of the paper is organized as follows. We first introduce the model and notation in Section II. In Section \ref{sec:ex}, we present an example that
sheds light on Theorem 7 of RWZ. In Section \ref{sec:id} we define the notion of redundant identifying restrictions and provide a modified necessary and sufficient condition for (exact) global identification. Section \ref{sec:exs} presents two examples based on influential papers in the SVAR literature. Section \ref{sec:conclusion} concludes.

\section{Model}
\label{sec:def}

We maintain the notation used in RWZ. Let $y_t$ be a $n\times 1$ vector of variables observed over the sample $t=1,\ldots, T$. The specification of the SVAR model is
\begin{equation}
\label{eq:SVAR}
        y_t^\prime A_0  = \sum_{l=1}^{p}y_{t-l}^\prime A_l + c +\e_t^\prime,
\end{equation}
where $\e_t$ is a $n\times 1$ multivariate normal white noise process with null expected value and covariance matrix equal to the identity matrix $I_n$. The $n\times n$ matrices $A_0,\,A_1,\ldots,\,A_p$ are the structural parameters and $c$ is a $1\times n$ vector of constant terms. The structural parameters are $(A_0,A_+)$, where $A_+^\prime\equiv (A_1^\prime,\ldots,\,A_l^\prime,\,c^\prime)$ is a $n\times m$ matrix with $m\equiv np+1$. We also assume that the initial conditions $y_1,\,\ldots,\, y_p$ are given and that $A_0$ is invertible. The set of structural parameters is denoted by $\ensuremath\mathbb{P}^S$, an open dense set of $\Re^{(n+m)n}$. The structural form can be written compactly as 
\begin{equation}
\label{eq:SVARc}
        y_t^\prime A_0  = x_{t}^\prime A_+ + \e_t^\prime
\end{equation}
where $x_t^\prime=\left(y_{t-1}^\prime,\ldots,\,y_{t-p}^\prime,\,1\right)$.

The reduced-form representation of (\ref{eq:SVARc}) is the standard VAR model,
\begin{equation}
\label{eq:VARc}
        y_t^\prime = x_{t}^\prime B + u_t^\prime,
\end{equation}
where $B=A_+A_0^{-1}$, $u_t^\prime=\e_t^\prime A_0^{-1}$, and $E(u_t\,u_t^\prime)=\Sigma=(A_0A_0^{\prime})^{-1}$. The reduced-form parameters are $(B,\,\Sigma)$, where $\Sigma$ is a symmetric and positive definite matrix. We denote the set of reduced-form parameters by 
$\ensuremath\mathbb{P}^R \subset \mathbb{R}^{nm+n(n+1)/2}$.
The relationship between the structural and reduced-form parameters is defined by the function 
$g:\ensuremath\mathbb{P}^S\rightarrow\ensuremath\mathbb{P}^R$, where $g\Atot=(A_+A_0^{-1},(A_0A_0^\prime)^{-1})$.

The definition of global identification is the standard one provided by \cite{Rothenberg71ECTA}; the absence of 
observationally equivalent parameters in the parametric space. We consider identification of the structural parameters by imposing zero restrictions on 
a transformation $f(\cdot)$ of the structural parameter space into the set of $k \times n$ matrices, $k \geq 1$, with domain $U\subset \ensuremath\mathbb{P}^S$. Such linear restrictions are represented by
\begin{equation}
\label{eq:restr}
        Q_j\fA e_j=0,\hspace{1cm}\text{for } \jn. 
\end{equation}
where $Q_j$ is a $k\times k$ selection matrix for $\jn$, and $e_j$ is the \textit{j}-th column of the $n\times n$ identity matrix $I_n$. 
The rank of $Q_j$ is denoted by $q_j$, which also represents the number of restrictions in the \textit{j}-th column of the transformed space
$\fA$. As in RWZ, we order the columns of $\fA$ according to 
\begin{equation}
\label{eq:ordering}
        q_1\geq q_2\geq\ldots\geq q_n. 
\end{equation}
We denote the set of orthonormal matrices by $\mathcal{O}(n)$ with generic element $P$.

According to RWZ, the transformation is admissible when the following condition holds.

\begin{condition}
\label{def:adm}
        The transformation $f(\cdot)$, with the domain $U$, is admissible if and only if for any $P\in\O$ and $\Atot\in U$, $f(A_0 P,A_+ P)=\fA P$.
\end{condition}

Moreover, RWZ impose the following two conditions when proving some of their results, the former of which is at the heart of the present paper. 

\begin{condition}
\label{def:reg}
        The transformation $f(\cdot)$, with the domain $U$, is regular if and only if $U$ is open and $f$ is continuously differentiable with $f^\prime \Atot$
        of rank $kn$ for all $\Atot \in U$.
\end{condition}

\begin{condition}
\label{def:streg}
        The transformation $f(\cdot)$, with the domain $U$, is strongly regular if and only if it is regular and $f(U)$ is dense in the set of 
        $k\times n$ matrices.
\end{condition}

To fix the sign of structural shocks, we need to impose sign normalization rules. Following RWZ, we define them as follows:

\begin{defin}[Normalization rule]
\label{def:norm}
        A normalization rule can be characterized by a set $N\subset \ensuremath\mathbb{P}^S$ such that for any structural parameter point 
        $\Atot\subset \ensuremath\mathbb{P}^S$, there exists a unique $n\times n$ diagonal matrix $D$ with plus or minus ones along the diagonal 
        such that $(A_{0}D,A_+D)\in N$.
\end{defin}

We are now able to define the set of restricted structural parameters as
\begin{equation}
\label{eq:R}
        R = \Big\{\Atot\in U\cap N\,\Big|\,Q_j\fA e_j=0\text{ for }\jn\Big\}.
\end{equation}
Following RWZ, we consider the following definition of identification when discussing whether or not the imposed restrictions can globally identify the structural parameters. 

\begin{defin}[Exact identification]
\label{def:exact}
        Consider an SVAR with restrictions represented by $R$. The SVAR is exactly identified if and only if, for almost any reduced-form parameter 
        point $(B,\Sigma)$, there exists a unique structural parameter point $\Atot\in R$ such that $g\Atot=(B,\Sigma)$.\footnote{In this definition, if the set of structural parameters under the restrictions $R$ constrains the reduced-form parameters, the domain of the reduced-form parameters for which the almost-sure property is required is restricted to $\PRt \subset \ensuremath\mathbb{P}^R$, where $\tilde{\mathbb{P}}^R$ is the set of reduced-form parameters generated by the structural parameters satisfying $R$. For instance, if $f(\cdot)$ maps
        the structural parameters to long-run impulse responses, its domain $U$ restricts the reduced-form VARs to being invertible. Then, $\tilde{\mathbb{P}}^R$ corresponds to the set of reduced-form parameters constrained to invertible VARs.}
\end{defin}

\section{An illustrative example}
\label{sec:ex}

In the setting described in the previous section, RWZ show a variety of necessary and sufficient conditions for the identifying restrictions $R$ with admissible $f(\cdot)$ to globally identify the structural parameters. Among those, the necessary and sufficient condition for exact identification presented in Theorem 7 of RWZ is the simplest and most attractive in practice, as it reduces verification of exact identification to a simple exercise of computing the rank of the matrices $Q_j$, $1 \leq j \leq n$. So that our exposition is self-contained, we present Theorem 7 of RWZ here:    

\bigskip

\noindent \textbf{Theorem 7 in RWZ}: \textit{Consider an SVAR with admissible and strongly regular restrictions represented by $R$.\footnote{Admissible and strongly regular restrictions represented by $R$ mean $f(\cdot)$ in (\ref{eq:R}) is admissible and strongly regular.} The SVAR is exactly identified if and only if $q_j = n - j$ for $1 \leq j \leq n$. }

\bigskip

The first result in this paper is that the assumption in Condition \ref{def:reg}, that looks like a rather technical one, and thus off consideration by practitioners,
can invalidate the ``if'' statement of Theorem 7 in RWZ, as shown by the following example.
 
\subsection{An example}
\label{sec:contr}

Consider a trivariate SVAR characterized by the following restrictions
\begin{equation}
\label{eq:ExRestr}
\begin{array}{ccc}
A_0 = \left(\begin{array}{ccc}
a_{11} & a_{12} & a_{13}\\
0 & a_{22} & a_{23}\\
0 & a_{32} & a_{33}
\end{array}\right) & \text{ and } &
IR_0 = \left(\begin{array}{ccc}
\times & 0 & \times\\
\times & \times & \times\\
\times & \times & \times
\end{array}\right)
\end{array}
\end{equation}
where $IR_0=(A_0^{-1})^{\prime}$ is the contemporaneous impulse response matrix, the symbol `$\times$' indicates that no restriction is imposed, and `0' represents
a zero (or exclusion) restriction.
The function $\fA$ will be\footnote{A more parsimonious way to define the transformation function could be to select just a subset of rows in the matrix in Eq. (\ref{eq:Exf}). This  alternative way, however, would invalidate the previous Condition \ref{def:adm}. For this reason we define $\fA$ by stacking the entire $A_0$ and $IR_0$ matrices.}
\begin{equation}
\label{eq:Exf}
        \fA=\left(\begin{array}{c}A_0\\IR_0\end{array}\right)=
        \left(\begin{array}{ccc}
        a_{11} & a_{12} & a_{13}\\
        0 & a_{22} & a_{23}\\
        0 & a_{32} & a_{33}\\
        \times & 0 & \times\\
        \times & \times & \times\\
        \times & \times & \times
        \end{array}\right).
\end{equation}
The matrices of restrictions defined in (\ref{eq:restr}) can be specified as
\begin{equation}
\label{eq:QEx}
        \begin{array}{lcr}
        Q_1 = \left(\begin{array}{cccccc}
        0 & 1 & 0 & 0 & 0 & 0\\
        0 & 0 & 1 & 0 & 0 & 0\\
        \hdashline[2pt/2pt]
        0 & 0 & 0 & 0 & 0 & 0\\
        0 & 0 & 0 & 0 & 0 & 0\\
        0 & 0 & 0 & 0 & 0 & 0\\
        0 & 0 & 0 & 0 & 0 & 0
        \end{array}\right), & \hspace{2cm} &
        Q_2 = \left(\begin{array}{cccccc}
        0 & 0 & 0 & 1 & 0 & 0\\
        \hdashline[2pt/2pt]
        0 & 0 & 0 & 0 & 0 & 0\\
        0 & 0 & 0 & 0 & 0 & 0\\
        0 & 0 & 0 & 0 & 0 & 0\\
        0 & 0 & 0 & 0 & 0 & 0\\
        0 & 0 & 0 & 0 & 0 & 0
        \end{array}\right).
        \end{array}
\end{equation}

According to Theorem 7 in RWZ, the SVAR is exactly (globally) identified, as the rank of the restriction matrices follows $q_1=n-1=2$, $q_2=n-2=1$ and $q_3=n-3=0$. However, analytical investigation shows the current set of identifying restrictions fails to achieve global identification. 

Let us express the reduced-form covariance matrix and its Cholesky decomposition as
\begin{equation}
\label{eq:ExSigma}
\begin{array}{ccc}
\Sigma = \left(\begin{array}{ccc}
\sigma_{11} & \sigma_{21} & \sigma_{31}\\
\sigma_{21} & \sigma_{22} & \sigma_{32}\\
\sigma_{31} & \sigma_{32} & \sigma_{33}
\end{array}\right) & \Rightarrow &
\Sigma_{tr} = \left(\begin{array}{ccc}
l_{11} & 0 & 0\\
l_{21} & l_{22} & 0\\
l_{31} & l_{32} & l_{33}
\end{array}\right).
\end{array}
\end{equation}
Imposing triangularity on $A_0$, we can obtain $A_0$ and $IR_0=A_0^{-1\prime}$ as
\begin{equation}
\label{eq:ExPar}
\begin{array}{ccc}
A_0^\prime = \Sigma_{tr}^{-1} = \left(\begin{array}{ccc}
\frac{1}{l_{11}} & 0 & 0\\
-\frac{l_{21}}{l_{11}l_{22}} & \frac{1}{l_{22}} & 0\\
\frac{l_{21}l_{32}-l_{22}l_{31}}{l_{11}l_{22}l_{33}} & -\frac{l_{32}}{l_{22}l_{33}} & \frac{1}{l_{33}}
\end{array}\right) & \Rightarrow &
IR_{0}= A_0^{-1\prime} = \left(\begin{array}{ccc}
l_{11} & 0 & 0\\
l_{21} & l_{22} & 0\\
l_{31} & l_{32} & l_{33}
\end{array}\right).
\end{array}
\end{equation}
Consider applying Algorithm 1 in RWZ to determine an orthogonal matrix $P$ that maps the $\Atot$ parameters under triangularity to the one satisfying the imposed restrictions. 

First, $\fA$  is
\begin{equation}
\label{eq:Exfnum}
        \fA=\left(\begin{array}{c}A_0\\IR_0\end{array}\right)=
        \left(\begin{array}{ccc}
        \frac{1}{l_{11}}  & -\frac{l_{21}}{l_{11}l_{22}} & \frac{l_{21}l_{32}-l_{22}l_{31}}{l_{11}l_{22}l_{33}}\\
        0 & \frac{1}{l_{22}} & -\frac{l_{32}}{l_{22}l_{33}}\\
        0 & 0 & \frac{1}{l_{33}}\\
        l_{11} & 0 & 0\\
        l_{21} & l_{22} & 0\\
        l_{31} & l_{32} & l_{33}
\end{array}\right).
\end{equation}
As in RWZ, let $\bar{Q}_1$ and $\bar{Q}_2$ be the matrices of indicators for the restricted elements of $\fA$ obtained by removing the row vectors of zeros from $Q_1$ and $Q_2$. Algorithm 1 in RWZ suggests calculating
\begin{equation}
\label{eq:Q1t}
        \tilde{Q}_1=\bar{Q}_1\fA=\left(\begin{array}{ccc}0 & \frac{1}{l_{22}} & -\frac{l_{32}}{l_{22}l_{33}}\\
                                                        0 & 0 & \frac{1}{l_{33}}\end{array}\right),
\end{equation}
and finding a unit-length vector that is orthogonal to the row vectors of $\tilde{Q}_1$. The QR decomposition of $\tilde{Q}_1$ and a sign normalization lead to $p_1 = (1,0,0)'$ as a unique unit vector satisfying $\tilde{Q}_1 p_1 = 0$, so we can pin down the first column vector of $P$.  

Next, to find the second column vector $p_2$ of $P$, we form the matrix
\begin{equation}
\label{eq:Q2t}
        \tilde{Q}_2=\left(\begin{array}{c}\bar{Q}_2\fA\\p_1^\prime\end{array}\right)
                                                 =\left(\begin{array}{ccc}l_{11} & 0 & 0\\ \hdashline[2pt/2pt] 1 & 0 & 0\end{array}\right)
\end{equation}
and search for a unit vector $p_2$ satisfying $\tilde{Q}_2\,p_2=0$. Since the rank of $\tilde{Q}_2$ is one for any value of $l_{11}$, we cannot pin down a unique $p_2$ (up to the sign normalization). From a geometric point of view, any vector belonging to the unit circle 
in $\Re^3$ orthogonal to the unit vector $p_1=\left(\begin{array}{ccc}1 & 0 & 0\end{array}\right)^\prime$ is admissible as $p_2$. This implies that given any reduced-form parameter value of $\Sigma$, the imposed restrictions fail to pin down a unique orthogonal matrix $P$, implying that, contrary to the claim in Theorem 7 of RWZ, global identification does not hold in this example. 

Some packaged algorithms for the QR decomposition, including the Matlab function $qr(\cdot)$, yield an orthogonal vector $p_2$ irrespective of whether it is unique or not. That is, if $\tilde{Q}_{2}$ is not full-rank, these algorithms implicitly select one unit vector $p_2$ from infinitely many admissible ones. As a result, an application of the ``if'' statement of Theorem 7 and naive implementation of Algorithm 1 in RWZ may fail to detect the failure of global identification and mislead subsequent impulse response analysis.

\subsection{Analytical investigation}
\label{sec:analytic}

A correct implementation of the RWZ procedure for checking global identification would have suggested an immediate warning, being the transformation function 
$\fA$ not regular. In fact, the domain of $\fA$ is the set of $n\times n$ matrices $A_0$ (being $IR_0=A_0^{-1}$), while its domain has dimension $2n\times n$.
According to Condition \ref{def:reg}, the transformation $f(\cdot)$, in order to be regular, must have the first derivative of rank $kn$. In the specific example,
the rank should be $nk = 2n^2$. However, the first derivative $f^\prime$ will be of dimension $2n^2\times n^2$, and thus the rank will be at most $n^2$, that is 
clearly less than $2n^2$. The identification scheme proposed in the example, thus, cannot be investigated by the RWZ procedure, although it is not so unrealistic from
an empirical point of view. In fact, combining zero restrictions on the contemporaneous relations among the endogenous variables and on the response on impact to
structural shocks could be a quite standard strategy for practitioners. In the following sections we will introduce a strategy to check for global identification even for models not admissible according to Condition \ref{def:reg}.

However, to understand why in this particular case the ``if'' statement of Theorem 7 of RWZ breaks down and how it can be modified, it is useful to determine analytically the special feature of the identifying restrictions specified in Eq. (\ref{eq:ExRestr}). 

We begin with the inversion of the $A_0$ matrix; the determinant of $A_0$ is
\begin{equation}
\label{eq:det}
        |A_0| = a_{11}a_{22}a_{33}+a_{12}a_{23}a_{31}+a_{13}a_{21}a_{32}-a_{13}a_{22}a_{31}-a_{11}a_{23}a_{32}-a_{12}a_{21}a_{33}
\end{equation}
and the adjunct matrix is
\begin{equation}
\label{eq:Adj}
                \text{Adj}(A_0) = \left(\begin{array}{ccc}
                a_{22}a_{33}-a_{32}a_{23} & -(a_{12}a_{33}-a_{32}a_{13}) & a_{12}a_{23}-a_{22}a_{13}\\
                -(a_{21}a_{33}-a_{31}a_{23}) & a_{11}a_{33}-a_{31}a_{13} & -(a_{11}a_{23}-a_{21}a_{13})\\
                a_{21}a_{32}-a_{31}a_{22} & -(a_{11}a_{32}-a_{31}a_{12}) & a_{11}a_{22}-a_{21}a_{12}
                \end{array}\right).
\end{equation}
The inverse is $A_0^{-1}=|A_0|^{-1}\,\text{Adj}(A_0)$. Substituting the two zero restrictions on $A_0$, $a_{21}=0$ and $a_{31}=0$, into $A_0^{-1\prime}$ leads to
\begin{equation}
\label{eq:AinvR}
                A_0^{-1\prime} = \frac{1}{a_{11}(a_{22}a_{33}-a_{23}a_{32})}\left(\begin{array}{ccc}
                a_{22}a_{33}-a_{32}a_{23} & 0 & 0\\
                -(a_{12}a_{33}-a_{32}a_{13}) & a_{11}a_{33} & -(a_{11}a_{32}-a_{31}a_{12})\\
                a_{12}a_{23}-a_{22}a_{13} & -a_{11}a_{23} & a_{11}a_{22}
                \end{array}\right)=IR_0.
\end{equation}
It is evident that the two restrictions on $A_0$ imply two zero restrictions on $IR_0$, $(A_{0}^{-1 \prime})_{[1,2]} = (A_{0}^{-1 \prime})_{[1,3]}=0$. One of these, $(A_{0}^{-1\prime})_{[1,2]}=0$, is exactly the zero restriction specified for $IR_0$ in (\ref{eq:ExRestr}). 
In other words, we intended to impose the three restrictions, but the two imposed on $A_0$ imply the third imposed on $IR_0$, 
so this third restriction was redundant. Due to this redundancy, the third restriction does not further constrain the admissible orthonormal matrix $P$, which translates into rank deficiency of $\tilde{Q}_2$. 

Although this redundancy phenomenon can occur in some realistic applications,\footnote{Many
influential empirical papers combine restrictions on both contemporaneous relationships among the endogenous variables 
and the contemporaneous impulse responses. Examples include \cite{Blanchard89}, \cite{BlanPerotti02}, \cite{Bernanke86}.} 
whether or not any of the imposed set of restrictions are redundant cannot be directly assessed by the simple necessary and sufficient condition in Theorem 7 of RWZ. As a way to uncover such redundancy, one may want to examine how a set of zero restrictions imposed on one structural object translates to zero restrictions on other objects. In Section IV below, we modify the necessary and sufficient condition of Theorem 7 of RWZ by offering a systematic way to detect redundancy of the imposed identifying restrictions.

\subsection{Detecting the failure of global identification}
\label{sec:DetectFail}

In their Theorem 6, RWZ provides an alternative necessary and sufficient condition for exact identification of SVARs. If we properly take into account that the imposed zero restrictions imply zero restrictions on other objects, this alternative approach can correctly detect a lack of global identification. We illustrate how in our example. 

For $1 \leq j \leq n$ and any $k \times n$ matrix $X$, let $M_j(X)$ be a $(k+j) \times n$ matrix defined by
\begin{equation*}
	M_j(X) = \begin{pmatrix} Q_j X \\ I_{j \times j} \mspace{20mu} O_{j \times (n-j)} \end{pmatrix},
\end{equation*}
where $Q_j$ is a $k \times k$ matrix defined in (\ref{eq:restr}). Theorem 6 of RWZ provides a necessary and sufficient condition for exact identification through the rank conditions for $M_j\big(f(A_0,A_+)\big)$. 

\bigskip	

\noindent \textbf{Theorem 6 in RWZ}: \textit{Consider an SVAR with admissible and strongly regular restrictions represented by $R$. The SVAR is exactly identified if and only if the total number of restrictions is equal to $n(n-1)/2$ and for some $(A_0,A_+) \in R$, $M_j\big(\fA\big)$ is of rank $n$ for $1 \leq j \leq n$.}

\bigskip 
   
As for the necessary and sufficient condition in Theorem 7, the condition in Theorem 6 is correct, but is limited to admissible and strongly regular transformations
$f(\cdot)$. As we have just seen here before, the analysis could be extended to a broader class of transformations, where the technical assumption on the full rank
of the first derivative of $f(\cdot)$ can be substituted by the less stringent and more practical condition of non-redundancy. Although the utilization of Theorem 6
in RWZ is precluded by the failure of the regularity Condition \ref{def:reg}, it is useful to understand what happens in our example and whether the condition can 
be recovered even in the case of non regular transformations. 

In the current example, the total number of restrictions imposed is 3 and it meets the condition for the total number of restrictions with $n=3$. We hence focus on checking the rank condition for $M_j\big(\fA\big)$, $j=1,2,3$.
In this check, we substitute the following matrices into $\fA$:
\begin{equation}
\label{eq:ExRestrA}
\begin{array}{ccc}
A_0 = \left(\begin{array}{ccc}
a_{11} & a_{12} & a_{13}\\
0 & a_{22} & a_{23}\\
0 & a_{32} & a_{33}
\end{array}\right) & \text{ and } &
IR_0 = \left(\begin{array}{ccc}
\times & 0 & 0\\
\times & \times & \times\\
\times & \times & \times
\end{array}\right),
\end{array}
\end{equation}
where the symbol `$\times$' denotes the parameters in Eq. (\ref{eq:AinvR}). We obtain, if $a_{22}a_{33}-a_{32}a_{23} \neq 0$, 
\begin{equation}
\label{eq:RkExA}
        \begin{array}{lcl}
        M_1\big(\fA\big)=\left(\begin{array}{ccc}
        0 & a_{22} & a_{23}\\
        0 & a_{32} & a_{33}\\
        \hdashline[2pt/2pt]
        1 & 0 & 0\end{array}\right) & \hspace{2cm} & \rk(M_1)=3\\&&\\
        M_2\big(\fA\big)=\left(\begin{array}{ccc}
        a_{22}a_{33}-a_{32}a_{23} & 0 & 0\\
        \hdashline[2pt/2pt]
        1 & 0 & 0\\
        0 & 1 & 0\end{array}\right) & \hspace{2cm} & \rk(M_2)=2<3.
        \end{array}
\end{equation}
Hence, the rank condition of Theorem 6 in RWZ fails. This is consistent with the conclusion in our analysis above; the imposed restrictions uniquely pin down the first column vector of $P$, but not the second column vector of $P$. Thus, plugging in the expression of $\fA$ obtained analytically under the imposed restrictions, the rank condition of Theorem 6 of RWZ correctly detects the failure of global identification due to the redundancy among the imposed identifying restrictions. 

It is important to note that understanding analytically the whole set of constraints implied by the imposed restrictions is crucial to correctly performing the check of the rank condition in Theorem 6 of RWZ. For instance, in the current example, if we were not aware of the redundancy issue of the identifying restrictions and incorrectly let the $(1,3)$-element of $M_2\big(\fA\big)$ be an unknown potentially nonzero free parameter, we would have erroneously claimed that $M_2\big(f(A_0, A_+)\big)$ were of rank 3 and concluded that the exact identification holds. If the dimension of the SVAR is large,  exhaustively investigating and figuring out the entire set of constraints implied by the zero restrictions on $f(A_0,A_+)$ is challenging. In such a case, immediate implementation of the rank conditions of Theorem 6 of RWZ is limited.

\subsection{Other examples of transformation}
\label{sec:OtherEx}

In the previous illuminating example we have considered zero restrictions on $A_0$ and $IR_0$. The next two examples introduce VARs, of potential interest in empirical applications, where the transformation functions are not regular, and thus Theorem 7 of RWZ cannot apply.

\begin{ex}[Short-run zero restrictions on a monetary policy shock] 
\label{ex:Lags}
	In practically all triangular quarterly SVARs,\footnote{See, among many others, \cite{CEE05JPE} and \cite{BG06RESTATS}.} it is assumed a monetary policy shock to have no impact within the quarter it hits 
	the economy. In a monthly VAR, it does correspond 	to zero restrictions on $IR_0$, $IR_1$, $IR_2$ and $IR_3$ of the monetary policy shock on indicators of the real economy, like industrial production.
	If the DGP is, for instance, a VAR with just two lags, the regularity assumption in Condition \ref{def:reg} is no longer valid. To see this, it is sufficient to write the 
	transformation function as
	\begin{equation}
	\label{eq:ExLags}
		f(A_0,A_1,A_2) = \Big(IR_0^\prime,IR_1^\prime,IR_2^\prime,IR_3^\prime\Big)^\prime,\nonumber
	\end{equation}
	with $IR_h = \Big(A_0^{-1}J^\prime F^h J\Big)^\prime$, where 
	\begin{equation}
	\label{eq:ExLagsMat}
		F = \left(\begin{array}{cc}A_1A_0^{-1} & I_n \\ A_2A_0^{-1} & 0\end{array}\right) 
		\hspace{1cm}\mathrm{and} \hspace{1cm} 
		J = \left(\begin{array}{c}I_n \\ 0\end{array}\right).\nonumber
	\end{equation}
	The dimension of $f(A_0,A_1,A_2)$ is $4n\times n$. Then, the rank of $f^\prime(A_0,A_1,A_2)$ is at most $3n$, that is strictly less than $4n$.
	Condition \ref{def:reg}, thus, is not satisfied. 
\end{ex}

\begin{ex}[Short-run, long-run, cumulated long-run restrictions] 
\label{ex:Cumul}
	Consider an SVAR model mixing aggregate demand and supply shocks \'{a} la \cite{BlanQuah1}, as well as permanent productivity shocks as in \cite{KPSW91}.
	The VAR model, as in \cite{KPSW91}, can be made of a) non-stationary $I(1)$ (potentially cointegrated) variables, where the long-run dynamics is governed by a 
	combination of both permanent and transitory structural shocks; b) first differences of non-stationary $I(1)$ variables, and c) stationary variables. 
	The set of constraints, thus, could be represented by zero restrictions on the response of impact $IR_0$, on the long-run responses $IR_{\infty}$, defined as
    \begin{equation}
	\label{eq:IRinfty}
		IR_{\infty} = \Bigg(A_0^\prime-\sum_{l=1}^{p} A_l^\prime\Bigg)^{-1},\nonumber
	\end{equation}
	and on the cumulated impulse responses ($L_{\infty}^c=\sum_{h=0}^{\infty}IR_h$). 
	In the same vain as in the Example \ref{ex:Lags}, if the estimated reduced-form specification is characterized by just one lag, Condition \ref{def:reg} is no longer valid
	and the RWZ's machinery cannot be implemented.
\end{ex}

\section{Extending the necessary and sufficient condition for exact identification}
\label{sec:id}

In this section we provide a modified necessary and sufficient condition for exact identification that eliminates the redundancy issue erroneously invalidating Theorem 7 of RWZ and extend the condition to a broader class of transformation functions. Our proposal relies on the sequential feature of Algorithm 1 in RWZ and checks the rank condition for uniqueness of the $j$-th column vector $p_j$ for each $j = 1, \dots, n$.  
 
Given the reduced-form parameters $(B, \Sigma)$, choose $(A_0, A_+)$ to be an unrestricted set of structural parameters satisfying $\Sigma = (A_0^{\prime})^{-1} (A_0)^{-1}$ and $B=A_+ A_0^{-1}$, such as $A_0^{\prime} = \Sigma_{tr}^{-1}$ and $A_+ = B (\Sigma_{tr}^{-1})'$. Let
\begin{equation}
\label{eq:Qt}
\tilde{Q}_1 = Q_1 f(A_0, A_+),  \ \text{and} \      \tilde{Q}_j=\left(\begin{array}{c}
        Q_j\fA\\p_1^\prime\\\vdots\\p_{j-1}^\prime
        \end{array}\right) \ \text{for $j=2,\ldots,n$}.
\end{equation}
By Theorem 5 and Algorithm 1 of RWZ, the exact identification of SVARs follows if and only if, for almost every reduced-form parameters $(B, \Sigma)$, the orthogonality conditions $\tilde{Q}_jp_j=0$ combined with the sign normalization restrictions pin down a unique orthogonal matrix $P$. 

For $P$ to be uniquely determined, it is necessary to have $q_j = n-j$ for all $1 \leq j \leq n$. This is, however, not a sufficient condition, because if any of the orthogonal vectors $(p_1, \dots, p_{j-1})$ is linearly dependent on the row vectors of $Q_j f(A_0,A_+)$, a rank-deficient $\tilde{Q}_j$ fails to pin down a unique 
$p_j$. This is exactly the mechanism that caused the systematic failure of global identification in our illustrative counterexample. To rule out such rank-deficiency in the characterization of the global identification condition,  we introduce the following concept:

\begin{defin}[Non-redundant restrictions]
\label{def:RedRes}
    Given reduced-form parameter $(B,\Sigma)$, let $A_0^{\prime} = \Sigma_{tr}^{-1}$ and $A_+ = B (\Sigma_{tr}^{-1})'$. Identifying restrictions for an SVAR that
		are represented by zero restrictions $Q_j\fA e_j=0$, $j=1, \dots, n$, are \textit{non-redundant} at given reduced-form parameter point, $(B,\Sigma)$ if for
		every $j=2, \dots, n$, orthogonal vectors $(p_1, \dots, p_{j-1})$ are linearly independent of the row vectors of $Q_j\fA$, i.e., $\tilde{Q}_j$ defined in 
		Eq. (\ref{eq:Qt}) is full row-rank for all $j=2, \dots, n$.
\end{defin}

If the imposed zero restrictions are non-redundant and the rank condition of Theorem 7 in RWZ holds, we can guarantee 
\begin{equation}
\label{eq:RedRes}
			\rk(\tilde{Q}_j) = \rk \left(\begin{array}{c}Q_j\fA\\p_1^\prime\\\vdots\\p_{j-1}^\prime\end{array}\right)=n-1
\end{equation}
for all $j=1, \dots, n$. We can therefore solve for an orthonormal matrix $P$ uniquely by sequentially solving $\tilde{Q}_j p_j = 0$, for $j=1, \dots, n$. If non-redundancy of the imposed restrictions holds for almost any reduced-form parameter point $(B,\Sigma)$, we can achieve exact identification. We hence obtain the following theorem that modifies Theorem 7 of RWZ.

\begin{theo}[A necessary and sufficient condition for exact identification]
\label{theo:NecSuffCond}
        Consider an SVAR with admissible restrictions represented by $R$. The SVAR is exactly identified at the point $(A_0,A_+)\in R$ if and only if 
         $q_j=n-j$ for $\jn$ and the restrictions are non-redundant at $(A_0,A_+)$.
\end{theo}

\begin{proof}
		Let the first $j-1$ shocks be identified. It means that the quantity
		\begin{equation}
		\label{eq:Pj}
			\setlength{\dashlinegap}{2pt}
			P_{j-1} = \left[\begin{array}{c:c:c:c}
							p_1 & p_2 & \cdots & p_{j-1}
						\end{array}\right]\nonumber
		\end{equation}
		is uniquely determined. If the \textit{j}-th shock is not identified, instead, then $p_j$ and $p_j^*$, by construction orthogonal to $P_{j-1}$,
		are both admissible, with $p_j^*$ any unit-length rotation of $p_j$.
		If $p_j$ is admissible, then
		\[
			Q_j\fA\,p_j=0
		\]
		and, similarly, if $p_j^*$ is admissible, then
		\[
			Q_j\fA\,p_j^*=0.
		\]
		In other words, both $p_j$ and $p_j^*$ belong to the null space of $Q_j\fA$, that, by assumption, is of dimension $j$.
		Moreover, by construction, they are orthogonal to $P_{j-1}$, too. The two vectors, thus, must be contained in the intersection of the two
		null spaces of $Q_j\fA$ and $P_{j-1}$, that is equivalent to the null space of the $(n-j)+(j-1)\times n$ matrix
		\[
			\tilde{Q}_j = \left(\begin{array}{c} 
											Q_j\fA\\ P_{j-1}^\prime
										\end{array}\right).
		\]
		Being $p_j$ and $p_j^*$ linearly independent, the null space of $\tilde{Q}_j$ must be at least of dimension two. 
		Now, using the rank-nullity theorem, we can say that $\rk\,\tilde{Q}_j\leq n-2$, and, thus, cannot be full. 
		This proves the sufficiency of the rank condition in the Theorem \ref{theo:NecSuffCond}.
		
		Proving the other direction of the condition is very simple. In fact, if the rank of $\tilde{Q}_j$ is equal to $n-1$ (that means that the 
		restrictions are not redundant), then the orthogonal complement of $\tilde{Q}_j$ will have dimension equal to one. Thus, there will be just two
		candidates for the vector $p_j$, but being one opposite to the other, only one will be retained. If this happens for every $j=\{1,\ldots,n\}$,
		there will be just one orthogonal matrix $P$ transforming the parameter point $(A_0,A_+)\in U$ into the parameter point $(A_0P,A_+P)\in R$,
		that satisfies the restrictions. Finally, the fact that this last result does correspond to global identification can be proved by using 
		Theorem 5 in RWZ, that does not require the restrictions to be neither regular nor strongly regular.
\end{proof}

In comparison to Theorem 7 of RWZ, our Theorem \ref{theo:NecSuffCond} adds the non-redundancy condition of the imposed restrictions as a part of necessary and sufficient condition. 
Importantly, removing Conditions \ref{def:reg} and \ref{def:streg} from the assumptions precludes the nice results in Theorem 3 in RWZ, saying that if the model is identified
in one point of $R$, it is identified almost everywhere in $R$. Our result, instead, is specific to the single point we are interested in. Building on and modifying Algorithm 1 of RWZ, the next algorithm correctly judges if exact identification
holds or not at the specific point $\big(B,\,\Sigma\big)\in\,R$, that could be the ML estimation, or any draw from the posterior distribution of the reduced-form parameters 
in a Bayesian VAR.

\begin{algo} 
\label{algo:Exact}
   Consider an SVAR with admissible restrictions represented by $R$ that satisfy $q_j=n-j$, for $\jn$. 
   Let $\big(B,\,\Sigma\big)$, be any reduced-form parameters.

   Perform the following steps:
        
     \begin{enumerate}
        \item Let $A_0^{\prime} = \Sigma_{tr}^{-1}$ and $A_+ = B (\Sigma_{tr}^{-1})'$, where $\Sigma_{tr}$ is the lower-triangular Cholesky factor of $\Sigma$. 
        \item Sequentially check the rank conditions for non-redundancy, i.e., check if $\rk(\tilde{Q}_j) = n-1$ holds, 
                where $\tilde{Q}_1=Q_1\fA$ and 
				\begin{equation}
                \label{eq:QtAlgo}
									\Qtj=\left(\begin{array}{c}Q_j\fA\\p_1^\prime\\\vdots\\p_{j-1}^\prime\end{array}\right)
                \end{equation}
                for $j = 2, \dots, n$, and $p_j$ is an $n \times 1$ unit-length vector satisfying $\tilde{Q}_j p_j = 0$ which is unique (up to sign normalization) if 
								$\mathrm{rank}(\tilde{Q}_{j})=n-1$ holds for all preceding $j=1, \dots, j-1$. 
     \end{enumerate}
  If the reduced-form parameter point passes Step 2 of the current algorithm, we conclude that the imposed identifying restrictions $R$ achieve exact identification. 
	 If not, we conclude that the imposed identifying restrictions do not achieve exact identification. 
\end{algo} 

The construction of orthonormal vectors $p_1, \dots, p_n$ by solving $\tilde{Q}_j p_j = 0$ sequentially for $j=1, \dots, n$, as incorporated in Step 2 of Algorithm \ref{algo:Exact}, is proposed in  Algorithm 1 of RWZ. For the purpose of checking exact identification, the important feature of our algorithm is the step of checking $\mathrm{rank}(\tilde{Q}_j) = n-1$ for all $j=1, \dots, n$. This extra step, which is absent in Algorithm 1 of RWZ, is necessary to detect failure of exact identification due to redundancy of the identifying restrictions.  


\section{Application}
\label{sec:exs}

In this section we show a couple of examples that, in two different dimensions, emphasize the impact of our proposal. The former is based on an influential work by \cite{Gali92}, who investigates the US economy through an SVAR model with different kinds of restrictions. The latter, instead, investigates global identification in proxy-SVAR models, with an application to fiscal policy shocks.\footnote{The code and data can be downloaded from the following repository:
\href{https://drive.google.com/file/d/18vGcSimY9-g1rLO4ECrVl6qdatlE8Z-6/view?usp=drive_link}{codes and datasets}.}

\subsection{An IS-LM model for the US economy}
\label{sec:ISLM}

\cite{Gali92} is a seminal contribution in the macroeconomic literature that attempts to understand the dynamics of US macroeconomic variables in the postwar period. Within the augmented IS-LM framework, he identifies four shocks as main drivers of the fluctuations observed in some macroeconomic variables: money supply, money demand, IS and aggregate supply shocks. The identification of such shocks is obtained as a combination of zero restrictions on the structural parameters of an SVAR model for the following variables: log of real GNP ($y_t$); short term interest rate ($i_t$); log of CPI ($p_t$) and log of M1 ($m_t$). Given the time series properties of the variables, the author comes up with the following transformations, that do represent the information set in the VAR model:
\begin{equation}
    \label{eq:GaliVariables}
    x_t = \big(\Delta y_t,\,\Delta i_t,\,i_t-\Delta p_t,\,\Delta m_t-\Delta p_t\big)^\prime\nonumber
\end{equation}
that is treated as a covariance stationary stochastic process.

\begin{table}[htbp]
\centering
\caption{Identifying Restrictions}
\label{tab:gali_identification}
\begin{tabular}{clc}
\hline\hline
& \vspace{0.1cm} &\\
& Long-run restrictions & \\
\hline\\
& R1: no long-run effects of money supply shocks on GNP & \\
& R2: no long-run effects of money demand shocks on GNP & \\
& R3: no long-run effects of IS shocks on GNP & \\
& & \\
& Short-run restrictions & \\
\hline\\
& R4: no contemporaneous effect of money supply shocks on output & \\
& R5: no contemporaneous effect of money demand shocks on output & \\
& R6: contemporaneous prices do not enter the money supply rule &\\
\hline\hline
\end{tabular}
\end{table}

The identifying restrictions are described in Table \ref{tab:gali_identification}. In terms of the notation used in this paper, the restrictions directly involve elements on the parameters of the structural form ($A_0$), but also transformation of them, like the on-impact response ($IR_0=A_0^{-1}$) and the long-run response ($IR_\infty$). Specifically, the restrictions can be written in terms of the transformation function $f(A_0,A_+)$ as
\begin{equation}
    \label{eq:Gali_f}
    f(A_0,A_+) = \left(\begin{array}{c}A_0\\IR_0\\IR_\infty\end{array}\right) 
               = \begin{array}{cc} & \begin{array}{cccc} \varepsilon_t^{ms} & \varepsilon_t^{md} & \varepsilon_t^{is} & \varepsilon_t^s\end{array}\\
                 \begin{array}{c}\Delta y_t\\\Delta i_t\\i_t-\Delta p_t\\\Delta m_t-\Delta p_t\\\Delta y_t\\\Delta i_t\\i_t-\Delta p_t\\\Delta m_t-\Delta p_t\\\Delta y_t\\\Delta i_t\\i_t-\Delta p_t\\\Delta m_t-\Delta p_t\\\end{array} &
                 \left[\begin{array}{cccc}
                      a_{11} & a_{12} & a_{13} & a_{14}\\
                      a_{21} & a_{22} & a_{23} & a_{24}\\
                      a_{31} & a_{32} & a_{33} & a_{34}\\
                      -a_{31}& a_{42} & a_{43} & a_{44}\\
                      \hdashline
                      0      & 0      & \times & \times\\
                      \times & \times & \times & \times\\
                      \times & \times & \times & \times\\
                      \times & \times & \times & \times\\
                      \hdashline
                      0      & 0      & 0      & \times\\
                      \times & \times & \times & \times\\
                      \times & \times & \times & \times\\
                      \times & \times & \times & \times\\
                 \end{array}\right]
                 \end{array}
\end{equation}
The $f(A_0,A_+)$ matrix is of dimension $(k\times n)$, where here $k=3n$ and $n=4$. In order to see whether the function is regular, we need to calculate the first derivative of $f(A_0,A_+)$ and determine its rank, which should be equal to $nk=3n^2$ for all $(A_0,A_+)\in U$.  
 
Clearly, for the same reason as in the previous Examples \ref{ex:Lags} and \ref{ex:Cumul} the rank is deficient for any $(A_0,A_+)\in U$ if the VAR is characterized by just one lag. Moreover, independently of the number of lags, it is not so simple to calculate the first derivative of a rather complicate matrix function like $f(A_0,A_+)$ in Eq. (\ref{eq:Gali_f}) and to establish whether its rank is full for all the elements in its domain. 

What we offer in this paper is a simple check that works even if the transformation function is not regular and that can be easily implementable in an algorithm that already represents a fundamental step in the estimation process.

The selection matrices $\bar{Q}_1$, $\bar{Q}_2$, $\bar{Q}_3$ and $\bar{Q}_4$, obtained from the respective $Q_j$ matrices, $j=1,\ldots,4$, by removing the zero rows, are defined as
\begin{eqnarray}
    \bar{Q}_1 & = & \left(\begin{array}{cccccccccccc}
                        0	&	0	&	1	&	1	&	0	&	0	&	0	&	0	&	0	&	0	&	0	&	0	\\
                        0	&	0	&	0	&	0	&	1	&	0	&	0	&	0	&	0	&	0	&	0	&	0	\\
                        0	&	0	&	0	&	0	&	0	&	0	&	0	&	0	&	1	&	0	&	0	&	0
                    \end{array}\right)\nonumber\\[10pt]
    \bar{Q}_2 & = & \left(\begin{array}{cccccccccccc}
                        0	&	0	&	0	&	0	&	1	&	0	&	0	&	0	&	0	&	0	&	0	&	0	\\
                        0	&	0	&	0	&	0	&	0	&	0	&	0	&	0	&	1	&	0	&	0	&	0	\\
                    \end{array}\right)\nonumber\\[10pt]
    \bar{Q}_3 & = & \left(\begin{array}{cccccccccccc}
                        0	&	0	&	0	&	0	&	0	&	0	&	0	&	0	&	1	&	0	&	0	&	0	\\
                    \end{array}\right).\nonumber
\end{eqnarray}
The fourth shock, $\varepsilon_t^s$, presenting no restrictions, is characterized by an empty matrix $\bar{Q}_4=\left(\,\,\right)$.

Here we implement our Algorithm \ref{algo:Exact} to check for the global identification in the original SVAR presented in \cite{Gali92}. We first estimate through maximum likelihood a gaussian VAR(4) as in Eq. (\ref{eq:VARc}) for the sample 1955-87 and obtain the parameters $B$ and $\Sigma$. In Step 1, then, we fix $A_0^\prime=\Sigma_{tr}^{-1}$ and $A_+ = B(\Sigma_{tr}^{-1})^\prime$. According to Step 2, for $j=1$, we have 
\begin{equation}
    \label{eq:GaliQt1}
    \tilde{Q}_1 = \bar{Q}_1 f(A_0, A_+) =
                  \left(\begin{array}{cccc} 
                    0.119	&	0.145	&	0.766	&	1.183	\\
                    0.801	&	0.113	&	0.065	&	-0.136	\\
                    1.261	&	-0.020	&	-0.095	&	0.116	
                  \end{array}\right),\nonumber
\end{equation}
that, being $\rk (\tilde{Q}_1)=n-1=3$, allows to globally identify the first structural shock $\varepsilon_t^{ms}$ and derive $p_1^\prime=\big(0.043,\,-0.845,\,0.487,\,-0.217\big)^\prime$. When $j=2$, we have
\begin{equation}
    \label{eq:GaliQt2}
    \tilde{Q}_2 = \left(\begin{array}{c}\bar{Q}_2 f(A_0, A_+)\\p_1^\prime\end{array}\right) =
                  \left(\begin{array}{cccc} 
                    0.801	&	0.113	&	0.0645	&	-0.136	\\
                    1.261	&	-0.020	&	-0.095	&	0.116	\\
                    0.043	&	-0.845	&	0.487	&	-0.217	
                  \end{array}\right),\nonumber
\end{equation}
for which $\rk (\tilde{Q}_2)=n-1=3$, allowing us to obtain $p_2=\big(0.005,\,0.274,\,0.745,\,0.608\big)^\prime$. For $j=3$, we have
\begin{equation}
    \label{eq:GaliQt3}
    \tilde{Q}_3 = \left(\begin{array}{c}\bar{Q}_3 f(A_0, A_+)\\p_1^\prime\\p_2^\prime\end{array}\right) =
                  \left(\begin{array}{cccc} 
                    1.261	&	-0.020	&	-0.095	&	0.116	\\
                    0.043	&	-0.845	&	0.487	&	-0.217	\\
                    0.005	&	0.274	&	0.745	&	0.608	
                  \end{array}\right)\nonumber
\end{equation}
which has $\rk (\tilde{Q}_3)=n-1=3$, leading to derive $p_3=\big( -0.111,\,-0.459,\,-0.449,\,0.759\big)^\prime$. Finally, for $j=4$, 
\begin{equation}
    \label{eq:GaliQt4}
    \tilde{Q}_4 = \left(\begin{array}{c}p_1^\prime\\p_2^\prime\\p_3^\prime\end{array}\right) =
                  \left(\begin{array}{cccc} 
                    0.043	&	-0.845	&	0.487	&	-0.217	\\
                    0.005	&	0.274	&	0.745	&	0.608	\\
                    -0.111	&	-0.459	&	-0.450	&	0.759	
                  \end{array}\right)\nonumber
\end{equation}
that, as before, features $\rk (\tilde{Q}_4)=n-1=3$, and allows to derive $p_4=\big(0.993,\,-0.016,\,-0.075,\,0.091\big)^\prime$. The orthogonal matrix mapping the parameters of the reduced form to the structural ones is obtained by collecting the $p_j$ vectors as $P=\big(p_1,\,p_2,\,p_3,\,p_4\big)$. Algorithm \ref{algo:Exact}, thus, allows us to establish that the SVAR model in \cite{Gali92} is globally identified at the maximum likelihood estimate $\big(B,\,\Sigma\big)$, regardless of whether the rather complicated regularity condition holds.

\subsection{Global identification in proxy-SVARs with an application to fiscal policy shocks}
\label{sec:ProxySVAR}

In this section we deal with proxy-SVARs and show how our procedure can be used to check whether the model is globally identified or not. As an example, we use a modified version of the seminal \cite{BlanPerotti02} SVAR, where a combination of zero restrictions and external proxies are used to identify the structural parameters.  

We define the proxy-SVAR model as in \cite{GKR19}, where $y_t$ is the $n\times 1$ vector of endogenous variables and $m_t$ the $k\times 1$ vector of external instruments. Let the SVAR as in Eq. (\ref{eq:SVAR}) be complemented by a model for the proxies as follows 
\begin{equation}
    \label{eq:proxySVAR}
    m_t^\prime\Gamma_0 = \varepsilon_t^\prime \Lambda + m_{t-1}^\prime \Gamma_1 +\ldots+ m_{t-\ell}^\prime \Gamma_\ell + \nu_t^\prime  
\end{equation}
where $\Gamma_0,\Gamma_1,\ldots,\Gamma_\ell$ are $k\times k$ matrices of parameters, with $\Gamma_0$ invertible, and $\Lambda$ is a $n\times k$ matrix. Moreover, we assume that $\big(\varepsilon_t^\prime,\nu_t^\prime\big)^\prime|\mathcal{F}_{t-1} \sim N\big(0_{(n+k)\times 1},\,I_{n+k}\big)$, with $\mathcal{F}_{t-1}$ denoting the information set at time $t-1$, which includes the lags of $y_t$ and $m_t$. Eq. (\ref{eq:proxySVAR}) is an SVAR model for the external instruments which also emphasizes that $E\big(m_t\varepsilon_t^\prime\big)=\Gamma_0^{-1\prime}\Lambda^\prime$. The validity of the proxies $m_t$ suggests they are correlated with $k$ structural shocks (\textit{relevance}), say, for simplicity, the last ones, while uncorrelated with the remaining $n-k$ structural shocks (\textit{exogeneity}).

Substituting the SVAR in Eq. (\ref{eq:SVARc}) into Eq. (\ref{eq:proxySVAR}) leads to
\begin{equation}
    m_t^\prime \Gamma_0 = y_t^\prime  A_0 \Lambda - x_t^\prime A_+ \Lambda  + m_{t-1}^\prime \Gamma_1 + \ldots + m_{t-\ell}^\prime \Gamma_\ell + \nu_t^\prime\nonumber
\end{equation}
from which
\begin{eqnarray}
    m_t^\prime & = & y_t^\prime A_0\Lambda \Gamma_0^{-1} - x_t^\prime A_+\Lambda\Gamma_0^{-1} + m_{t-1}^\prime \Gamma_1 \Gamma_0^{-1} + \ldots + m_{t-\ell}^\prime \Gamma_\ell \Gamma_0^{-1} + \nu_t^\prime                             \Gamma_0^{-1}\nonumber\\[10pt]
               & = & y_t^\prime D + x_t^\prime G + m_{t-1}^\prime H_1 + \ldots + m_{t-\ell}^\prime H_\ell + v_t^\prime\label{eq:VARproxy}
\end{eqnarray}
that is the reduced form specification for the external instruments, where we defined $D=A_0\Lambda\Gamma_0^{-1}$, $G=- A_+\Lambda\Gamma_0^{-1}$, $H_i=\Gamma_i\Gamma_0^{-1}$, $i=1,\ldots,\ell$, and $v_t^\prime = \nu_t^\prime \Gamma_0^{-1}$. The definition of the models for $y_t$ and $m_t$, as well as the assumptions on the two terms $\varepsilon_t$ and $\nu_t$, lead to the definition of the validity of the instruments as
\begin{equation}
    \label{eq:validity}
    E\big(m_t\varepsilon_t^\prime\big)=\Gamma_0^{-1\prime}\Lambda^\prime=D^\prime A_0^{-1\prime}=\Big(0_{k\times (n-k)}\,,\,\Psi\Big)
\end{equation}
where $\Psi$ is a full-rank $k \times k$ matrix, and where we have exploited the fact that $\Lambda \Gamma_0^{-1} = A_0^{-1}D$. Eq. (\ref{eq:validity}) simply states that linear combinations of the first $n-k$ columns of $A_0^{-1\prime}$ are equal to zero. Such linear combinations are given by the $k$ rows of $D^\prime$. Put differently, we can say that the validity of the external proxies introduces $k\times (n-k)$ zero restrictions on the function $f(A_0,A_+)$ that can be added to other zero restrictions to identify the SVAR. As for standard SVAR models, Theorem \ref{theo:NecSuffCond} and Algorithm \ref{algo:Exact} can be used to study exact identification in proxy-SVARs. \cite{AF19} provide a condition for the identification of proxy-SVARs (see their Proposition 2), but restricting to local identification, only, and to zero restrictions on the on-impact responses. Our condition, thus, is more general in terms of restrictions and focuses on the global identification of the model.  

As an empirical implementation, we consider the seminal contribution by \cite{BlanPerotti02} on fiscal policy shocks. The two authors, in a small-scale SVAR, identify government spending and tax shocks and characterize their dynamic effects on US activity in the postwar period. Identification is reached by a combination of zero restrictions on the parameters of the model and on an event study approach allowing the authors to calibrate the elasticities to output of government purchases and of taxes minus transfers. Although not directly discussed by the authors, a formal investigation, reported in Appendix \ref{app:BP}, shows that the \cite{BlanPerotti02} SVAR is globally identified. 

In what follows, we study the global identification of fiscal policy shocks by exploiting recent developments in proxy SVARs. Instead of calibrating some elasticities, we propose using external information to identify, at least, one of the two fiscal shocks. We consider the same set of endogenous variables as in the \cite{BlanPerotti02} setup, i.e.,
\begin{equation}
    \label{eq:BP_SVAR}
    y_t = \big(t_t,\,x_t,\,g_t\big)^\prime,
\end{equation}
where the three variables are logarithms of quarterly tax revenues, GDP, and public spending, respectively. The three associated structural shocks are $\varepsilon_t=\big(\varepsilon_t^t,\,\varepsilon_t^x,\,\varepsilon_t^g\big)$, i.e., the tax policy, the output shock and the fiscal spending shock, respectively. Moreover, we consider two potential fiscal instruments for the two fiscal shocks: the \cite{MR2014}'s series of unanticipated tax shocks identified through a narrative analysis of tax policy decisions ($p_t^t$); a series of unanticipated fiscal spending shocks recently introduced in \cite{ACCF23} ($p_t^g$). Given the potential lack of the exogeneity condition for the tax instrument, as advocated by \cite{Lewis21} and \cite{KKP26}, we consider only the fiscal spending instrument for our identification purposes and fix $m_t = \big(p_t^g\big)$. Similar results, in terms of identification, can be obtained by considering the tax instrument. 

The validity condition states that, other than being correlated with the shock of interest $\varepsilon_t^g$, the proxy must be uncorrelated with the first two shocks, i.e. the tax shock $\varepsilon_t^t$ and the output shock $\varepsilon_t^x$. As seen before, this restriction corresponds to $D^\prime A_0^{-1\prime} = \big(0_{1\times 2}\,,\,\Psi\big)$, where $\Psi$ is a non-zero scalar. Moreover, following \cite{BlanPerotti02}, we can add a further restriction stating that public spending does not enter the tax equation, i.e. $a_{3,1}=0$. The restricted transformation matrix, thus, will assume the form
\begin{equation}
    \label{eq:BP_f}
    f(A_0,A_+) = \left(\begin{array}{c}A_0\\IR_0\end{array}\right) 
               = \begin{array}{cc}
                 & \begin{array}{ccc} \varepsilon_t^{t} & \varepsilon_t^{x} & \varepsilon_t^{g}\end{array}\\
                 \begin{array}{c}t_t\\x_t\\g_t\\t_t\\x_t\\g_t\end{array}
                 &
                \left[
                \begin{array}{ccc}
                a_{11} & a_{12} & a_{13}\\
                a_{21} & a_{22} & a_{23}\\
                0      & a_{32} & a_{33}\\
                \hdashline
                \tikz[remember picture,baseline=(current bounding box.center)]\node (a) {$\times$}; & \tikz[remember picture,baseline=(current bounding box.center)]\node (c) {$\times$}; & \times\\
                \times & \times & \times\\
                \tikz[remember picture,baseline=(current bounding box.center)]\node (b) {$\times$}; & \tikz[remember picture,baseline=(current bounding box.center)]\node (d) {$\times$}; & \times
                \end{array}
                \right]
                \end{array}
\end{equation}
\begin{tikzpicture}[overlay,remember picture]
\node[draw,rectangle,fit=(a)(b),inner sep=3pt] {};
\end{tikzpicture}
\begin{tikzpicture}[overlay,remember picture]
\node[draw,rectangle,fit=(c)(d),inner sep=3pt] {};
\end{tikzpicture}
where the vertical rectangles select the elements of the linear combination, given by $D$, that must be equal to zero. Important, the RWZ methodology cannot be used to check for the global identification, as the transformation function includes both $A_0$ and its inverse $IR_0=A_0^{-1\prime}$, as in the example in Section \ref{sec:contr}.

Using US data for the sample 1950:Q1 to 2006:Q4, we check whether the proxy SVAR for $y_t = \big(t_t,\,x_t,\,g_t\big)^\prime$, with $m_t = \big(p_t^g\big)$, is globally identified or not. All variables are per capita, deflated by the GDP deflator, and expressed in logarithms. As in \cite{MR2014}, \cite{CK17} and \cite{AFN25}, the VAR specification includes four lags, a constant and a linear trend. A zero restriction on the simultaneous matrix $A_0$, complemented by the two zero restrictions provided by the exogeneity of the external instrument, are those described in Eq. (\ref{eq:BP_f}), where the transformation function $f(A_0,A_+)$ has as arguments $A_0$ and the transpose of its inverse $IR_0=A_0^{-1\prime}$. 

The estimation of the reduced form of the proxy SVAR model described in Eq.s (\ref{eq:VARc}) and (\ref{eq:VARproxy}) can be easily performed through maximum likelihood, obtaining the parameters $B$, $\Sigma$ and $D$. Specifically, we obtain
\[
    \hat{D} = \big(-0.0011\,,\,0.0042\,,\,1.0466\big)^\prime
\]
which constitutes an essential element for fixing the restrictions. Thus, the selection matrices $\bar{Q}_1$, $\bar{Q}_2$ and $\bar{Q}_3$, obtained, as before, from the respective $Q_j$ matrices, $j=1,\ldots,3$, by removing the zero rows, are defined as
\begin{eqnarray}
    \bar{Q}_1 & = & \left(\begin{array}{cccccccccccc}
                        0	&	0	&	1	&	0	&	0	&	0	\\
                        0	&	0	&	0	&	-0.0011 & 0.0042 & 1.0466
                    \end{array}\right)\nonumber\\[10pt]
    \bar{Q}_2 & = & \left(\begin{array}{cccccccccccc}
                        0	&	0	&	0	&	-0.0011 & 0.0042 & 1.0466\\
                    \end{array}\right)\nonumber
\end{eqnarray}
The third shock, $\varepsilon_t^g$, presenting no restrictions, is characterized by an empty matrix $\bar{Q}_3=\left(\,\,\right)$.
Following Step 1 of our Algorithm \ref{algo:Exact}, we fix $A_0^\prime=\Sigma_{tr}^{-1}$ and $A_+ = B(\Sigma_{tr}^{-1})^\prime$. According to Step 2, for $j=1$, we have 
\begin{equation}
    \label{eq:BPQt1}
    \tilde{Q}_1 = \bar{Q}_1 f(A_0, A_+) =
                  \left(\begin{array}{ccc} 
                       3.968  & -44.106 & 85.228\\
                       -0.00003 & 0.00003 & 0.01229
                  \end{array}\right),\nonumber
\end{equation}
that, being $\rk (\tilde{Q}_1)=n-1=2$, allows to globally identify the first structural shock $\varepsilon_t^{t}$ and derive $p_1^\prime=\big(0.996,\,0.094,\,0.002\big)^\prime$. When $j=2$, we have
\begin{equation}
    \label{eq:BPQt2}
    \tilde{Q}_2 = \left(\begin{array}{c}\bar{Q}_2 f(A_0, A_+)\\p_1^\prime\end{array}\right) =
                  \left(\begin{array}{ccc} 
                      -0.00003 & 0.00003 & 0.01229\\
                       0.996   & 0.094 & 0.002
                  \end{array}\right),\nonumber
\end{equation}
for which $\rk (\tilde{Q}_2)=n-1=2$, allowing us to obtain $p_2=\big(-0.094,\,0.996,\,-0.002\big)^\prime$. Finally, for $j=3$, 
\begin{equation}
    \label{eq:BPQt3}
    \tilde{Q}_3 = \left(\begin{array}{c}p_1^\prime\\p_2^\prime\end{array}\right) =
                  \left(\begin{array}{ccc} 
                        0.996 & 0.094 & 0.002\\
                        -0.094 & 0.996 & -0.002	\\
                  \end{array}\right)\nonumber
\end{equation}
that, as before, features $\rk (\tilde{Q}_3)=n-1=2$, and allows to derive $p_3=\big(-0.002,\,0.002,\,0.999\big)^\prime$. The orthogonal matrix mapping the parameters of the reduced form to the structural ones is obtained by collecting the $p_j$ vectors as $P=\big(p_1,\,p_2,\,p_3\big)$. Algorithm \ref{algo:Exact}, thus, allows us to establish that the fiscal policy proxy SVAR model is globally identified at the maximum likelihood estimate $\big(B,\,\Sigma,D\big)$, although the regularity condition does not hold.

\section{Conclusion}
\label{sec:conclusion}

Based on a counterexample, this note demonstrates the importance of some technical assumptions to be met for the sufficiency claim in Theorem 7 of RWZ, commonly used by applied macro-economists because of its simplicity, to be correctly implemented. Analytical investigation of this counterexample reveals the issue of redundancy among the identifying restrictions, which the rank conditions of Theorem 7 of RWZ, in the absence of such regulatory assumptions, overlook. To rectify this, and to open to larger sets of admissible parameter restrictions, we present, firstly, a new necessary and sufficient condition for exact identification and, secondly, a computational algorithm that can correctly detect redundant identifying restrictions and is easy to implement in practice. We recommend this procedure to any researcher who wishes to check global identification of SVARs and proxy SVARs under their choice of equality identifying restrictions.


\newpage
\renewcommand{\thesection}{\Alph{section}}\setcounter{section}{0}

\setlength{\baselineskip}{12pt} 
\bibliographystyle{ecta}
\bibliography{SVAR_global}
\addcontentsline{toc}{section}{Bibliography}

\newpage


\appendix


\newpage
\section{Appendix: Global Identification of the Blanchard and Perotti (2002)'s fiscal policy SVAR} 
\label{app:BP}

\citeauthor{BlanPerotti02}'s (\citeyear{BlanPerotti02}) paper represents one of the main contributions on 
the dynamic effect of fiscal policy shocks on the real economy. The authors consider a three-dimensional vector 
in the logarithms of quarterly taxes, spending, and GDP, all in real, per capita,
terms, $Y_t=(T_t,\,G_t,\,X_t)$. The equations of the model are 
\begin{eqnarray}
t_{t} &=& a_1x_t+a_2\epsilon_t^g+\epsilon_t^t\nonumber\\
g_{t} &=& b_1x_t+b_2\epsilon_t^t+\epsilon_t^g\nonumber\\
x_{t} &=& c_1t_t+c_2g_t+\epsilon_t^x\nonumber
\end{eqnarray}
where $u_t=(t_t,\,g_t,\,x_t)$ is the vector of reduced-form residuals and $\epsilon_t^t$, $\epsilon_t^g$ and $\epsilon_t^x$ are the mutually uncorrelated structural shocks they are interested in, with variances $\sigma_t^2$, $\sigma_g^2$ and $\sigma_y^2$, respectively. The model presents relationships among both observable variables and latent structural shocks. Using the terminology of \cite{AG} and \cite{LutBook06}, it is an AB-SVAR, for which only criteria for local identification do exist.\footnote{See \cite{AG}, \cite{LutBook06}, \cite{Hamilton94}. Recently, \cite{BL18} have proposed an identification criterion that does not depend on the unknown structural parameters.}

In any case, as it is, the SVAR is clearly not identified, and two more restrictions are needed. In fact the authors propose
three further restrictions to solve the identification issue: 1) $b_1=0$, 2) $a_1=2.08$ and, finally 3) either $a_2=0$ or $b_2=0$. 
Under these restrictions, and in particular when we fix $b_2=0$ (the situation is very similar for the alternative $a_2=0$), after some simple algebra shown in the proof of the following Proposition \ref{prop:BP}, the model becomes
\begin{equation}
    \label{eq:BPmatrix}
    \left(\begin{array}{ccc}\alpha_{11}&0&0\\\alpha_{21}&\alpha_{22}&\alpha_{23}\\\alpha_{31}&\alpha_{32}&\alpha_{33}\end{array}\right)
    \left(\begin{array}{c}t_t\\g_t\\x_t\end{array}\right) = 
    \left(\begin{array}{c}\e_t^t\\\e_t^g\\\e_t^x\end{array}\right)
\end{equation}
where the new structural shocks $\e=(\e_t^t\,,\e_t^g\,,\e_t^x)^\prime$ have been normalized to have unit variances, $E(\e_t\,\e_t^\prime)=I_3$, the $(3\times 3)$ identity matrix, and furthermore, where the $\alpha_{ij}$ parameters are defined as
\begin{equation}
    \label{eq:BPrestr}
    \begin{array}{lclcl}
    \alpha_{11}=1/\sigma_g&&&&\\
    \alpha_{21}=-2.08/\sigma_t&,& \alpha_{22}=1/\sigma_t &,& \alpha_{23}=-a_2/\sigma_t\\
    \alpha_{31}=-c_1/\sigma_x&,& \alpha_{32}=-c_2/\sigma_x &,& \alpha_{33}=1/\sigma_x.
    \end{array}
    \end{equation}
The model, thus, presents two restrictions in the first equation, one restriction in the second, and no restriction in the third. However, while the two zero restrictions in the first equation allow identifying the orthogonal vector $p_1$, the restriction on the second equation takes the form $\alpha_{21}=-2.08\alpha_{22}$ and apparently is not in line with the set of restrictions mentioned in RWZ. However, with a bit of algebra, shown in the proof of the following proposition, it is possible to prove that this last restriction corresponds to a zero constraint on the second orthogonal vector $p_2$. All these results are summarized in the following proposition.

\newpage

\begin{prop}
\label{prop:BP}
The Blanchard-Perotti's fiscal policy AB-SVAR is globally identified.
\end{prop}

\begin{proof} 
The corresponding AB-SVAR representation of the Blanchard-Perotti's model is
\begin{equation}
    \label{eq:BP}
    \left(\begin{array}{ccc}1&0&-a_1\\0&1&-b_1\\-c_1&-c_2&1\end{array}\right)
    \left(\begin{array}{c}t_t\\g_t\\x_t\end{array}\right) = 
    \left(\begin{array}{ccc}1&a_2&0\\b_2&1&0\\0&0&1\end{array}\right)
    \left(\begin{array}{c}\e_t^t\\\e_t^g\\\e_t^x\end{array}\right).\nonumber
    \end{equation}
As it is, the SVAR is clearly not identified, and two more restrictions are needed. The authors propose
three further restrictions to solve the identification issue: 1) $b_1=0$, $a_1=2.08$ and, finally 3) either $a_2=0$ or $b_2=0$. 
Under these restrictions, and in particular when we fix $b_2=0$ (the situation is very similar for the alternative $a_2=0$), the model becomes
\begin{eqnarray}
    t_{t} &=& 2.08x_t+a_2\epsilon_t^g+\epsilon_t^t\nonumber\\
    g_{t} &=& \epsilon_t^g\nonumber\\
    x_{t} &=& c_1t_t+c_2g_t+\epsilon_t^x\nonumber
\end{eqnarray}
whose SVAR representation is
\begin{equation}
    \left(\begin{array}{ccc}1&0&-2.08\\0&1&0\\-c_1&-c_2&1\end{array}\right)
    \left(\begin{array}{c}t_t\\g_t\\x_t\end{array}\right) = 
    \left(\begin{array}{ccc}1&a_2&0\\0&1&0\\0&0&1\end{array}\right)
    \left(\begin{array}{c}\epsilon_t^t\\\epsilon_t^g\\\epsilon_t^x\end{array}\right)\nonumber.
\end{equation}
As standard in the simultaneous equation systems, the normalization has been obtained by imposing unit coefficients on the main diagonal of the matrix collecting the relationships between the observable variables, leaving the structural shocks to have unconstrained variances $\sigma_t^2$, $\sigma_g^2$ and $\sigma_y^2$, respectively. An alternative way, more familiar in the SVAR literature\footnote{For the details, see \cite{HWZ07}, \cite{WZ03}, or \cite{RWZ10}.}, instead, is to normalize all the equations in order to have unit-variance structural shocks $\e=(\e_t^t\,,\e_t^g\,,\e_t^x)^\prime$.

This different normalization, together with simple algebra, allow to rewrite the model in terms of the RWZ's representation. Moreover, ordering the equations according to the number of restrictions, leads to representations provided in previous Eq.s (\ref{eq:BPmatrix}) and (\ref{eq:BPrestr}). Using the algorithm in RWZ, the two zero restrictions in the first equation allow to pin down a unique unit vector $p_1$ that will constitute the first column of the orthogonal matrix $P$. The restriction on the second equation, instead, can be written as $\alpha_{21}=-2.08\alpha_{22}$ and, with very simple algebra, it is possible to show that
\begin{eqnarray}
    \alpha_{21}=-2.08\alpha_{22} & \Rightarrow & \left(\Chol^{-1}\,e_1\right)^\prime p_2 = -2.08\,\left(\Chol^{-1}\,e_2\right)^\prime p_2\nonumber\\
    & \Rightarrow & \left[\left(\Chol^{-1}\,e_1\right)^\prime + 2.08\,\left(\Chol^{-1}\,e_2\right)^\prime\right] p_2 = 0. \nonumber
\end{eqnarray}
It is thus possible to see that this further restriction continues to represent a standard zero restriction in line with RWZ, that allows to uniquely pin down the second vector $p_2\perp p_1$ composing the orthogonal matrix $P$. Finally, in $\Re^3$, there will be a unique unit vector $p_3\perp p_2\perp p_1$. This proves that there exists a unique orthogonal matrix $P=(p_1,\,p_2,\,p_3)$ transforming the unrestricted reduced-form parameters to the structural ones satisfying the imposed restrictions. The Blanchard-Perotti AB-SVAR model, thus, is globally identified.

\end{proof}

Although the literature has proposed only criteria for local-identification of AB-SVARs, the previous proposition shows that the restrictions imposed by Blanchard and Perotti are sufficient for having global identification, even in the case of calibrated coefficients. However, this last point is extremely delicate and, as emphasized in \cite{BK20}, can lead to local identification, although the \textit{sequential} number of restrictions, as shown in Theorem 7 of RWZ, might be in favor of global identification.

\end{document}